\pdfoutput=1 
\documentclass[onecolumn,11pt]{article}

\usepackage{geometry}
\usepackage{amsmath,amssymb,amsthm,booktabs,color,doi,graphicx,latexsym,url,xcolor,xspace,paralist,comment,hyperref}	
\usepackage{color, colortbl}

 \hypersetup{
   colorlinks=true,
   linkcolor=[rgb]{0,0.1,0.5},
   citecolor=[rgb]{0,0.1,0.5},
   urlcolor=[rgb]{0,0.1,0.5},
   urlcolor=[rgb]{0,0.1,0.5},
   pdftitle={--}, 
   pdfauthor={}
 }

\newtheorem{theorem}{Theorem} 

\theoremstyle{definition}
\newtheorem{definition}[theorem]{Definition}

\newcommand{\namedref}[2]{\hyperref[#2]{#1~\ref*{#2}}}

\newcommand{\equationref}[1]{\hyperref[#1]{Eq~(\ref*{#1})}}
\newcommand{\theoremref}[1]{\hyperref[#1]{Theorem~\ref*{#1}}}
\newcommand{\lemmaref}[1]{\hyperref[#1]{Lemma~\ref*{#1}}}
\newcommand{\remarkref}[1]{\hyperref[#1]{Remark~\ref*{#1}}}

\newcommand{\local}{\ensuremath{\mathsf{LOCAL}}\xspace}
\newcommand{\congest}{\ensuremath{\mathsf{CONGEST}}\xspace}

\newcommand{\supported}{\ensuremath{\mathsf{SUPPORTED}}\xspace}

\newcommand{\cc}{\ensuremath{\mathsf{CC}}\xspace}
\newcommand{\rcc}{\ensuremath{\mathsf{RCC}}\xspace}

\title{Brief Announcement: Does Preprocessing Help under Congestion?}
\date{}
\author{Klaus-Tycho Foerster$^*$ \quad Janne H. Korhonen$^\dagger$ \quad Joel Rybicki$^\dagger$ \quad Stefan Schmid$^*$\\ \\
\small$^*$Faculty of Computer Science, University of Vienna, Austria\\ \small$^\dagger$Institute of Science and Technology Austria}

\begin{document}

\maketitle

\begin{abstract}
This paper investigates the power of preprocessing in the \congest model.
Schmid and Suomela (ACM HotSDN 2013) introduced the \supported \congest model to study
the application of distributed algorithms in Software-Defined Networks (SDNs). 
In this paper, we show that a large class of lower bounds in the \congest model still hold in the \supported model, highlighting the robustness of these bounds. This also raises the question how much does preprocessing help in the \congest model.
%
%
\end{abstract}


\section{Introduction}\label{sec:introduction}
Common models of distributed computation typically consider scenarios
where the computation always starts from \emph{scratch}, i.e., in an \emph{unknown} communication topology. 
However, in many practical scenarios, the communication 
topology does not change as frequently as the problem input. 
For example, the distributed algorithm may always be run in networks whose topology is known in advance, but the input instance may vary. In such cases, it is natural to 
\emph{support} distributed algorithms by allowing preprocessing of the underlying network topology~\cite{hotsdn13}.

With this in mind, Schmid and Suomela~\cite{hotsdn13} proposed two \supported models of distributed computation to enhance distributed algorithms with the power of preprocessing: the \supported \local and \supported \congest. Subsequently, Korhonen and Rybicki considered subgraph detection problems~\cite{KR17} in the \supported \congest model, whereas Foerster et~al.~\cite{Foerster2019} investigated the power of the \supported \local model. In this paper, we focus on  \supported \congest.

\paragraph{Contribution.} 
We observe that many lower bounds in the standard \congest model still hold under such preprocessing.
Given that intuitively preprocessing seems to be very powerful, this may come as a surprise. %
This raises the question of how much preprocessing actually helps in the \congest model. 
Indeed, it may be either that the power of preprocessing is very limited or that the current lower bounds for non-supported \congest are not tight. In the light of this, we propose the following challenge: 
\emph{is there a separation between supported and non-supported models}?
If the answer is no, then there may be a way to easily simulate preprocessing, and thus, simplify algorithm design in the \congest model. In the converse case, preprocessing may offer a practical way to accelerate current distributed algorithms. 

\paragraph{Model.} 
In the \supported \congest model, the communication topology is an undirected graph $H=(V,E)$ and each node has a unique identifier of size $O(\log n)$ bits. The logical state is given by an undirected subgraph, i.e., the \emph{input} graph $G \subseteq H$, which inherits the identifiers in $H$. The computation proceeds in two steps: First, in the preprocessing phase, the nodes may compute any function on $f(H)$ and store the result locally. In the second phase, the nodes are tasked to solve a problem instance on the input graph $G$ in the \congest model. To this end, the edges of $H$ may be used for communication and additionally the local outputs of the preprocessing. Note that the congested clique~\cite{DBLP:conf/spaa/LotkerPPP03} model
is a special case of the \supported model: the support $H$ is simply a clique. 
%
In addition, one may also restrict the communication of the \supported model to just the input graph $G$ after 
preprocessing; this model is called the \emph{passive} \supported model.


%

\section{\boldmath Lower Bounds for the \supported \congest Model}

We show that \congest lower bounds obtained using the now standard \emph{family of lower bound graphs} construction~\cite{DBLP:journals/corr/abs-1901-01630} easily translate to the \supported model. To this end, we adapt here the proof of Abboud et al.~\cite{DBLP:journals/corr/abs-1901-01630}. Using existing constructions for families of lower bound graphs then immediately give the lower bounds shown in Table~\ref{table}. We note that this technique does not directly cover the lower bounds of Das Sarma et al.~\cite{dassarma12}, though we believe they can be similarly translated. 

\paragraph{Two-party communication complexity.}
Let $f \colon \{0,1\}^{2k} \to \{0,1\}$ be a Boolean function. In~the two-party communication game on $f$, there are two players who receive a private $k$-bit string $x_0$ and $x_1$ as input, and the task is to have at least one of the players compute $f(x) = f(x_0, x_1)$. The \emph{deterministic communication complexity} $\cc(f)$ of a function $f$ is the maximum number of bits the two players need to exchange in the worst case (over all deterministic protocols and input strings) in order to compute $f(x_0,x_1)$. Similarly, the \emph{randomised communication complexity} $\rcc(f)$ is the worst-case complexity of protocols, which 
compute $f$ with probability at least $2/3$ on all inputs.

\begin{definition}
    Let $f_n \colon \{ 0, 1 \}^{2k(n)} \to \{ 0, 1 \}$ and $C \colon \mathbb{N} \to \mathbb{N}$ be functions and $\Pi$ a graph predicate. Suppose that there exists a constant $n_0$ such that for all $n > n_0$ and $x_0,x_1 \in \{0,1\}^{k(n)}$ there exists a (weighted) graph $G(n,x_0,x_1)$ satisfying the following properties:
\begin{compactenum}
    \item $G(n, x_0,x_1)$ satisfies $\Pi$ if and only if $f_n(x_0,x_1)=1$,
    \item $G(n, x_0,x_1) = (V_0 \cup V_1, E_0 \cup E_1 \cup S)$, where
        \begin{compactitem}
            \item[--] $V_0$ and $V_1$ are disjoint and $|V_0 \cup V_1| = n$,
            \item[--] $E_i \subseteq V_i \times V_i$ for $i \in \{0,1\}$,
            \item[--] $S \subseteq V_0 \times V_1$ is a cut and has size at least $C(n)$, and
            \item[--] the (weighted) subgraph $G_i = (V_i, E_i)$ only depends on $i$, $n$ and $x_i$, i.e., $G_i = G_i(n, x_i)$.
        \end{compactitem}
\end{compactenum}
    If $\mathcal{G}(n) = \{ G(n,x) \colon x \in \{0,1\}^{2k(n)} \}$, then $\mathcal{F} = (\mathcal{G}(n))_{n > n_0}$ is a \emph{family of lower bound graphs}.
    
\end{definition}

\begin{theorem}\label{lower-bounds}
    Let $\mathcal{F}$ be a family of lower bound graphs. Any algorithm deciding $\Pi$ on a graph family $\mathcal{H}$ containing $\bigcup \mathcal{G}(n)$ for all $n > n_0$ in the passive or active \supported \congest model with bandwidth $b(n)$ needs
    \[
    \Omega\left( \frac{\cc(f_n)}{C(n) b(n)} \right) \quad \textrm{ and } \quad \Omega\left( \frac{\rcc(f_n)}{C(n) b(n)} \right)
    \]
    deterministic and randomised rounds, respectively.
\end{theorem}
\begin{proof}
    Suppose $\mathcal{A}$ is an algorithm that decides $\Pi$ on the graph family $\mathcal{H}$ in $T(n)$ communication rounds. We now construct a two-player protocol $\pi$ that computes $f_n(x_0, x_1)$ by simulating $\mathcal{A}$. Let $x_0,x_1 \in \{0,1\}^{k(n)}$ be the input and $G = G(n,x_0,x_1)$ and $H = \bigcup \mathcal{G}(n)$.
     
     Given its input $x_i$, player~$i$ can locally construct the subgraph $G_i(n,x_i) \subset G(n,x_0,x_1)$. Note that given $G_{i}(x_i)$, the support graph $H$ does not reveal any information about $E_{1-i}$ or $x_{1-i}$ to player~$i$, since for any $y \in \{0,1\}^{k(n)}$ we have $G_{1-i}(y) \subseteq H[V_{1-i}]$.
    Simulating any messages sent between vertices of $G_i(n,x)$ can be done without any communication with player $1-i$. Any messages from $V_{i}$ to $V_{1-i}$ must go across the cut $S$ and are communicated by player $i$ to player $1-i$. As in each round each player communicates at most $b(n)$ bits over any edge in $S$, the total amount of bits communicated during the course of $T$ rounds is at most $2b(n) |S(x)| T(n) \ge 2b(n)C(n)T(n)$,
    which must be at least $\cc(f_n)$ for deterministic algorithms and $\rcc(f_n)$ for randomised algorithms.
Thus, the claim follows by observing that $\frac{\cc(f_n)}{2b(n)C(n)} \le T(n)$ and $\frac{\rcc(f_n)}{2b(n)C(n)} \le T(n)$.
\end{proof}

\begin{table}[h]
\definecolor{lgray}{rgb}{0.9,0.9,0.9}

\begin{footnotesize}
\begin{center}
\begin{tabular*}{\linewidth}{@{}l@{\extracolsep{\fill}}l@{}l@{}}
Lower bound & Problem   \\
\midrule
\rowcolor{lgray}
$\Omega(n^{1/2}/\log n)$ & $4$-cycle~\cite{drucker13}, $2k$-cycle~\cite{KR17}, Girth ($(2 - \varepsilon)$-apx.)~\cite{frischknecht2012}\\
$\Omega(n/\log n)$ & $(2k+1)$-cycle~\cite{drucker13}, APSP, Diameter ($(3/2 - \varepsilon)$-apx.)~\cite{frischknecht2012}\\
\rowcolor{lgray}
$\Omega(n/(\log n)^2)$ & Diameter on sparse graphs~\cite{AbboudCK16} \\
$\Omega(n/(\log n)^3)$ & Diameter and radius ($(3/2 - \varepsilon)$-apx.), eccentricities ($(5/3 - \varepsilon)$-apx.), all on sparse graphs~\cite{AbboudCK16}\\
\rowcolor{lgray}
$\Omega(n^{2-1/k}/(k\log n))$ & Subgraph detection (for any $k$)~\cite{subgraph_spaa}\\
$\Omega(n^2/(\log n)^2)$ & Min.~vertex cover, max.~independent set, chrom.~number ($(4/3-\varepsilon)$-apx.), weighted 8-cycle~\cite{CHKP17}\\
\rowcolor{lgray}
$\Omega(n^2)$ & Identical subgraphs (det.\ only)~\cite{CHKP17}\\
\midrule
\end{tabular*}
\end{center}
\end{footnotesize}
\caption{Lower bounds that transfer from the \congest to the \supported \congest model.}\label{table}

\end{table}

\section{Towards New Algorithmic Opportunities?}

We saw in the last section that many lower bounds from the \congest model translate directly to the \supported \congest model, even though intuitively, the \supported model may seem significantly more powerful. 
This raises the question if the \supported model is actually a stronger model in a meaningful sense or if the prior lower bounds were so strong that they easily transferred.
\paragraph{First separation results.} Prior work on the \supported \local model~\cite{Foerster2019} already pointed out that computing an upper bound on the network size separates the \local ($\Omega(D)$ rounds) and the \supported \local model (0 rounds). 
Analogous results hold if the support graph is promised to have certain (monotone) properties that apply to all its subgraphs, e.g., being $k$-colorable.
These results directly carry over to the \supported \congest model, providing a $0$ vs $\Omega(D)$ round separation, even in an identifier-independent setting.
On the other hand, in the \supported \local model, all problems can be solved trivially in diameter time, but can the \supported \congest go further?
Observe that the na\"ive problem of collecting all identifiers does not provide an $\Omega(n^2)$ separation in the \supported \congest model, as the problem may only depend on the input graph, which may omit nodes present in the support graph.
Notwithstanding, we can alter the problem s.t.\ each node has to e.g.\ output two sets $I_0,I_1$ of identifiers with $|I_0|=|I_1|$, where $I_1$ contains 
a superset of all identifiers in $G$ and none of the identifiers in $I_0$ appear in $G$.
\paragraph{Open questions and possibilities.}
While the separation results for restricted graph classes can be directly used to accelerate many specialized algorithms (e.g., coloring when the support graph has a small chromatic number), we leave it as an open question how the \supported model can be leveraged outside the case of collecting identifiers and providing upper bounds on the graph size, even though the latter is sometimes needed as an input for some algorithms.
We believe that exciting possibilities arise, no matter the outcome to this open question.
For example, if the \congest model could simulate the \supported model with negligible overhead beyond the previously mentioned exceptions, the \supported model could greatly simplify algorithm design by incorporating preprocessing.
On the other hand, even a strong separation could lead to significantly faster algorithms in neighboring research areas, e.g.~for Software Defined Networks~\cite{hotsdn13}.


\bibliographystyle{clickable-doi-or-url}
\bibliography{support-congest}

\end{document}